\documentclass[11pt]{article}

\usepackage[hmargin=1.2in,vmargin=1.2in]{geometry}
\usepackage{amsmath,amsthm,amssymb,amsfonts,verbatim}
\usepackage{fullpage}
\usepackage{hyperref}

 \providecommand{\F}{\mathbb{F}}

\title{Beating the probabilistic lower bound on $q$-perfect hashing\thanks{Part of this work appeared at SODA 2021 \cite{XY21} where we only showed that the probabilistic lower bound can be improved for sufficiently large $q$ with $q\neq 2\pmod{4}$.}}
\author{{\sf Chaoping Xing}\thanks{School of Electronic Information and Electrical Engineering,  Shanghai Jiao Tong University, Shanghai 200240, China. Email:{\tt xingcp@sjtu.edu.cn}. The research of C. Xing is supported in part by the National Natural Science Foundation of China under Grant 12031011 and the National Key Research and Development Projects 2021YFE0109900 and 2020YFA0712300.}
\and {\sf Chen Yuan}\thanks{School of Electronic Information and Electrical Engineering,  Shanghai Jiao Tong University, Shanghai 200240, China. Email:{\tt chen\_yuan@sjtu.edu.cn}. The research of C. Yuan is supported in part by the National Natural Science Foundation of China under Grant 12101403.}
}

\newtheorem{lemma}{Lemma}[section]
\newtheorem{theorem}[lemma]{Theorem}
\newtheorem{prop}[lemma]{Proposition}
\newtheorem{con}[lemma]{Conjecture}
\newtheorem{cor}[lemma]{Corollary}

\newtheorem{defn}{Definition}

\theoremstyle{remark}
\newtheorem{rmk}{Remark}

\renewcommand{\epsilon}{\varepsilon}
\renewcommand{\le}{\leqslant}
\renewcommand{\ge}{\geqslant}

\def\proj{\mathrm{proj}}

\def\ZZ{\mathbb{Z}}

\def \mA {\mathcal{A}}

\def \mA {\mathcal{A}}

\def \mS {\mathcal{S}}

\newcommand{\Ga}{\alpha}

\def \bc {{\bf c}}
\def \bi {{\bf 1}}
\def \bx {{\bf x}}

\def \by {{\bf y}}

\def\proj{{\rm proj}}

\date{}
\begin{document}
\maketitle

\begin{abstract}
For an integer $q\ge 2$, a perfect $q$-hash code $C$  is a  block code over $[q]:=\{1,\ldots,q\}$ of length $n$ in which every subset $\{\mathbf{c}_1,\mathbf{c}_2,\dots,\mathbf{c}_q\}$ of $q$ elements is separated, i.e., there exists $i\in[n]$ such that $\{\mathrm{proj}_i(\mathbf{c}_1),\dots,\mathrm{proj}_i(\mathbf{c}_q)\}=[q]$, where $\mathrm{proj}_i(\mathbf{c}_j)$ denotes the $i$th position of $\mathbf{c}_j$. Finding the maximum size $M(n,q)$ of  perfect $q$-hash codes of length $n$, for given $q$ and $n$, is  a fundamental problem in combinatorics, information theory, and computer science. In this paper, we are interested in asymptotic behavior of this problem. Precisely speaking, we will focus on the quantity $R_q:=\limsup_{n\rightarrow\infty}\frac{\log_2 M(n,q)}n$.

A well-known probabilistic argument shows an existence lower bound on $R_q$, namely $R_q\ge\frac1{q-1}\log_2\left(\frac1{1-q!/q^q}\right)$ \cite{FK,K86}. This is still the best-known lower bound till now except for the case $q=3$  \cite{KM}. The improved lower bound of $R_3$ was discovered in 1988 and there has been no progress on the lower bound of $R_q$ for more than $30$ years. In this paper we show that this probabilistic lower bound can be improved for $q$ from $4$ to $15$ and all odd integers between $17$ and $25$, and \emph{all sufficiently large} $q$.
\end{abstract}

\section{Introduction}
Probabilistic method is widely used to prove the existence of an object
meeting a certain condition in theoretical computer science and extremal combinatorics. Instead of constructing such object explicitly, one only needs to prove that such object occurs with positive probability. This feature makes it a powerful tool in deriving lower bound. Moreover, in most cases, the lower bound provided by probabilistic method turns out to be the best. However, some exceptional examples occur such as the Gilbert-Varshamov bound in coding theory \cite{TVZ} and the probabilistic lower bound on perfect hash codes \cite{KM}. In this paper, we study lower bounds on perfect hash codes and compare them with the probabilistic lower bound. There are many applications of perfect hashing: for example, see \cite{AYZ}, \cite{TT}.

A perfect $q$-hash code $C\subseteq [q]^n$ is a $q$-ary code such that for every subset of $C$ containing $q$ codewords, there exists an coordinate where the $q$ codewords in this subset have distinct values. By convention, the rate of this $q$-hash code is defined as $R_C=\frac{\log_2 |C|}{n}$.

The existence of a perfect $q$-hash code gives rise to a perfect $q$-hash family. To see this, let $C$ be the whole universe and the projection of each coordinate be a hash function. Then, for any $q$ elements of this universe, there exists a hash function mapping them to distinct values. Another application of perfect $q$-hash code is the zero-error list decoding on certain channels. A channel can be thought of as a bipartite graph $(V;W;E)$, where $V$ is the set of channel inputs, $W$ is the set of channel outputs, and  $(w,v)\in E$ if on input $v$, the channel can output $w$.
The $q/(q-1)$ channel then is the channel with $V = W =\{0,1,\dots,q-1\}$, and $(v,w)\in E$ if and only if $v\neq w$.
If we want to ensure that the receiver can identify a
subset of at most $q-1$ sequences that is guaranteed to contain the transmitted sequence,  one can communicate via $n$ repeated uses of the channel using the perfect $q$-hash code. See \cite{E88, DGR17}
for more details.

In this paper, we only consider the asymptotic behavior of rates of perfect $q$-hash codes, namely, we focus on the quantity $R_q:=\limsup_{n\rightarrow\infty}\frac{\log_2 M(n,q)}n$, where $M(n,q)$ stands for the maximum size of perfect $q$-hash codes of length $n$.

The study of $R_q$ could be dated back to $80$s. There are a few works dedicated to the upper bound on $R_q$. Fredman and Koml\'{o}s \cite{FK} showed a general upper bound: $R_q\le \frac{q!}{q^{q-1}}$ for all $q\ge 2$.  Arikan \cite{Arikan} improved this bound for $q=4$, and then Dalai, Guruswami and Radhakrishnan \cite{DGR17} further improved the upper bound on $R_4$.
Recently, Guruswami and Riazanov \cite{GR} discovered a stronger bound for every $q\geq 4$.
Costa and Dalai \cite{CD20} show that it is possible to explicitly compute this improvement over the previous upper bound.
Fiore, Costa and Dalai \cite{FCD} further improved the bound for small $b$ and $k$.

 Although there are some works towards tightening the upper bound on $R_q$, there are very few results about lower bounds on $R_q$. A plain probabilistic argument shows the existence of perfect $q$-hash code with rate  $R_q\ge\frac1{q-1}\log_2\left(\frac1{1-q!/q^q}\right)$ \cite{FK,K86}. This is still the best-known lower bound till now except for the case $q=3$  for which K\"{o}rner and Matron \cite{KM} found that the concatenation technique could lead to perfect $3$-hash codes beating the probabilistic lower bound. The improvement on the lower bound on $R_3$ was discovered in 1988 and there has been no progress on lower bounds on $R_q$ for more than 30 years.
K\"{o}rner and Matron's idea is to concatenate  an outer code, an $9$-ary $3$-hash code with an inner code, a perfect $3$-hash code with size $9$. They further posed an open problem whether there exist perfect $q$-hash codes beating the random argument for every $q$. In this paper, we provide a partial and affirmative answer to this open problem. We show that there exist perfect $q$-hash codes beating the random argument for all sufficiently large $q$ with $q\neq 2 \pmod 4$. To complement this result, we also prove the existence of perfect $q$-hash code that could beat random result for small  $q$ from $4$ to $15$ and odd $q$ between $17$ and $25$, as well as many other odd integers between $27$ and $155$.
Our computer  search result together with asymptotic result suggests that our construction might beat the probabilistic lower bound for every integer $q$.

The main technique of this paper is a modified version of concatenation. Unlike K\"{o}rner and Matron's concatenation where both inner and outer codes must be separated,  we abandon this separateness of inner code at a cost of imposing a stronger requirement on the outer code. By relaxing the condition that the inner code is a perfect $q$-hash code, we have more freedom to construct the inner code. As a result, we are able to improve the lower bound on $R_q$.

Before explaining our technique in detail, let us recall the concatenation technique introduced by K\"{o}rner and Matron. A plain probabilistic argument can prove the existence of an $m$-ary outer code $C_1$ of length $n_1$ that is $q$-separated with $q\le m$, i.e., for every $q$-element subset of $C_1$ (a $q$-element set is a set of size $q$), there exists $i\in\{1,2,\dots, n_1\}$ such that elements of this $q$-subset are pairwise distinct at position $i$. Then, they construct a perfect $q$-hash code $C_2$ of length $n_2$ as the inner code. By concatenating $C_1$ with $C_2$ (see Lemma \ref{lm:concatenation} for detail), they obtain a perfect $q$-hash code of length $n_1n_2$. In this way, they managed to prove the existence of $3$-perfect code beating the probabilistic lower bound.

 In our concatenation, we make a trade-off between inner code and outer code by relaxing the condition on the inner code and imposing a stronger condition on the outer code. By taking a set $\mA$ consisting of some $q$-element subsets of $[m]$, we apply the probabilistic method to show the existence of an $m$-ary outer code $C_1$ such that, for every $q$-element subset $\{\bc_1,\bc_2,\dots,\bc_q\}$ of $C_1$, there exists $i$ such that $\{\proj_i(\bc_1),\proj_i(\bc_2),\dots,\proj_i(\bc_{n_1})\}\in \mA$, where $\proj_i(\bc_j)$ stands for the $i$th coordinate of $\bc_j$. Note that K\"{o}rner and Matron's concatenation only requires that there exists $i$ such that $\{\proj_i(\bc_1),\proj_i(\bc_2),\dots,\proj_i(\bc_{n_1})\}$ are pairwise distinct. In this sense, we extend their idea by confining $\{\proj_i(\bc_1),\proj_i(\bc_2),\dots,\proj_i(\bc_{n_1})\}$ to be one of the subset in $\mA$. If there is an inner code $C_2$ such that at least $|\mA|$ $q$-codewords subsets of $C_2$ are separated, we can concatenate $C_1$ with $C_2$ to obtain a perfect $q$-hash code. Now, it remains to look for suitable inner code $C_2$. One good candidate for the inner code is the Maximum Distance Separable (MDS) code. In this paper, we let $C_2$ to be an $[3,2]$-MDS code over a abelian group of size $q$. We then reduce determining the number of separated  $q$-element subsets of $C_2$ to determining the number of  $q$-element subsets of $C_2$ in which all three positions are separated. It turns out that the latter problem is equivalent to the following well-known combinatorial problem: determine the number $s_q$ of pairs $(\pi_1,\pi_2)$ of bijections $[q]\rightarrow \ZZ_q$ such that $\pi_1+\pi_2$ is a bijection of $\ZZ_q$ as well. By using exact values of $s_q$ for small $q$ or estimates for moderate $q$ from \cite{K07}, we prove our results for odd $q\leq 155$. The value of $s_q$ was determined asymptotically by Eberhard, Manners, and Mrazovic \cite{BMM}, and by using this result, as well as related work of Eberhard \cite{Ebe2017}, we prove our result for all sufficiently large $q$.

 There is an asymptotic result on $s_q$ for odd number $q$ \cite{BMM} which can be used to estimate the number of separated  $q$-element subsets of $C_2$. As a result, we are able to improve $R_q$ for large odd $q$.
 Recently, this combinatorial problem is further extended to abelian group $G$ with $\sum_{x\in G}x=0$ \cite{Ebe2017}.
 In fact, an even stronger result was proved which holds for $(\pi_1,\pi_2)$ of bijections such that $\pi_1+\pi_2+f$ is a bijection for some function $f: [q]\rightarrow \ZZ_q$ with $\sum_{i=1}^{q}f(i)=\sum_{x\in G}x$.
 Due to this result, we can also extend our result to improve $R_q$ for every large $q$.

 We further extend this $[3,2]$-MDS code result to a $[4,2]$-MDS code. It turns out that an $[4,2]$-MDS code over an abelian group of size $q$ could lead to an even better lower bound on $R_q$. Our main result is summarized below.
\begin{theorem}
For every integer $q$ with $q\neq 2\pmod 4$, one has a lower bound
$$
R_q\ge-\frac{1}{4(q-1)}\log_2\bigg(\left(1-\frac{q!}{q^q}\right)^4-\left(\frac{3q}{\sqrt{e}}+o(q)\right)\left(\frac{q!}{q^q}\right)^3\bigg).
$$
For every integer $q$ with $q=2\pmod 4$, one has a lower bound
$$
R_q\ge-\frac{1}{3(q-1)}\log_2\bigg(\left(1-\frac{q!}{q^q}\right)^3-\left(\frac{q}{2\sqrt{e}}+o(q)\right)\left(\frac{(q!)}{q^{q}}\right)^3\bigg).
$$
This rate outperforms  the probabilistic lower bound, $R_{q}\ge-\frac{1}{(q-1)}\log_2(1-\frac{q!}{q^q})$, for all sufficiently large $q$.
\end{theorem}

We note that the numerical results imply that the same construction also beat the probabilistic lower bound for small $q$. This leads to the following conjecture.
\begin{con}
For every integer $q$, there exists a perfect $q$-hash code beating the probabilistic lower bound. Moreover, such construction can be obtained via a concatenation code defined in
Theorem \ref{thm:3.13}, Theorem \ref{thm:3.14} and Theorem \ref{thm:newbound}.
\end{con}



This paper is organized as follows. In Section 2, we propose a new concatenation technique and derive a lower bound on $R_q$ in terms of the number of separated $q$-element subsets of the inner code. In Section 3, we provide several candidates for the inner code of our concatenation technique and estimate the number of separated $q$-element subsets for these candidates. By plugging this number into the lower bound in Section 2, we manage to prove that the probabilistic lower bound on $R_q$ with $q\neq 2\pmod 4$ can be improved in many cases.
In Section 4, we provide another candidate that can beat the probabilistic lower bound for $q=2 \pmod 4$. In Section 5, we provide a construction that is not based on linear code which can further improve the lower bound on $R_5$ and $R_7$.

\section{$\mA$-friendly codes and concatenation}
\subsection{Hash code}
A set containing $q$ elements is called a $q$-element set.
Assume that $m\ge q$, then a $q$-element subset $\{\bc_1,\bc_2,\dots,\bc_q\}$ of $[m]^N$ is called separated if there exists $i\in[N]$ such that $\proj_i(\bc_1),\dots,\proj_i(\bc_q)$ are pairwise distinct. 
If $q$ is a prime power, we denote by $\F_q$ the finite field with $q$ elements and let $\ZZ_m:=\ZZ/m\ZZ$ be the group of integers modulo $m$.

A subset $C$ of $[m]^N$ is called an $m$-ary code of length $N$. For an integer $q\le m$, an $m$-ary code $C$ of length $N$ is called an $m$-ary
$q$-hash code if every $q$-element subset of $C$ is separated. In particular, we say that $C$ is a perfect $q$-hash code if $m=q$.

We generalize the notion of $m$-ary
$q$-hash codes. Let $\binom{[m]}{q}$ denote the collection of  all $q$-element subsets of $[m]$. Let $\mathcal{A}$ be a subset of $\binom{[m]}{q}$ and let $C$ be a code in $[m]^N$. We say that a $q$-element subset $\{\bc_1,\ldots,\bc_q\}$ of $[m]^N$ is $\mathcal{A}$-friendly if there exists $i\in[N]$ such that
$\{\proj_i(\bc_1),\proj_i(\bc_2),\dots,\proj_i(\bc_q)\}\in \mathcal{A}$. Otherwise, we call $\{\bc_1,\ldots,\bc_q\}$ an $\mathcal{A}$-unfriendly subset. If every $q$-element subset of $C$ is $\mathcal{A}$-friendly, we say that $C$ is an $\mathcal{A}$-friendly code. In particular,
this definition coincides with an $m$-ary  $q$-hash code when $\mathcal{A}=\binom{[m]}{q}$.
\subsection{Random $\mathcal{A}$-friendly codes}
In this subsection, by applying a probabilistic argument, we prove the existence of $\mathcal{A}$-friendly codes.
\begin{lemma}\label{list-recovery-con}
Let $\mathcal{A}$ be a nonempty subset of $\binom{[m]}{q}$. Then there exists an $m$-ary $\mathcal{A}$-friendly code $C$ of length $N$ and size at least $\left\lceil\frac{M}{3}\right\rceil$ as long as
\begin{equation}\label{eq:1}
{M\choose q}\left(1-\frac{q!|\mathcal{A}|}{m^q}\right)^N\leq \frac{M}{2q}.\end{equation}
for fixed $q,m,|\mathcal{A}|$.
\end{lemma}
\begin{proof}
From \eqref{eq:1}, it is clear that $M\leq m^{(1-\epsilon)N}$ for some constant $\epsilon$ when $N$ is large enough.
We sample $M$ codewords $\bc_1,\ldots,\bc_M$ uniformly at random in $[m]^N$ with replacement.
The number of collisions is at most $M/6$. To see this, let $X_{i,j}$ be the $0,1$-random variable such that
$X_{i,j}=1$ if $\bc_i=\bc_j$ and $X_{i,j}=0$ otherwise. It is clear $P[X_{i,j}=1]=m^{-N}$.
It follows that $E[\sum_{1\leq i<j\leq M}X_{i,j}]=\binom{M}{2}m^{-N}\leq M/6$ due to the fact
that $M\leq m^{(1-\epsilon)N}$. Next, we bound the number of $q$-element sets from these $M$ codewords that are not $\mathcal{A}$-friendly.
Let us fix a $q$-element set $\{\bc_1,\ldots,\bc_q\}$ with $\bc_i=(c_{i,1},\ldots,c_{i,N})$.
For any $j\in [n]$, the probability that $\{c_{1,j},\ldots,c_{q,j}\}\in \mathcal{A}$ is $\frac{q!|\mathcal{A}|}{m^q}$
as $c_{i,j}$ is picked uniformly at random in $[m]$.
It follows that the probability that
$\{\bc_1,\ldots,\bc_q\}$ is $\mathcal{A}$-unfriendly is $(1-\frac{q!|A|}{m^q})^N$. There are at most $\binom{M}{q}$
$q$-element sets from $\{\bc_1,\ldots,\bc_M\}$. By union bound, the expected number of $\mathcal{A}$-unfriendly $q$-element sets
is at most ${M\choose q}\left(1-\frac{q!|\mathcal{\mathcal{A}}|}{m^q}\right)^N\leq \frac{M}{2q}$. Remove all the codewords that lie in
any of these $\mathcal{A}$-unfriendly $q$-element sets. Then, we remove at most $q\times \frac{M}{2q}=\frac{M}{2}$ codewords. According to our previous argument, there are at most $M/6$ collisions among these $M$ codewords. Remove these $M/6$ codewords and we obtain the $\mathcal{A}$-friendly code of size at least $\frac{M}{3}$. The desired result follows.
\end{proof}
\begin{rmk}
Note that in~\cite{KM}, the set $\mA$ is the collection of all $q$-element subsets of $[m]$. Thus, our random argument can be viewed as
a generalization of the argument in \cite{KM}. This generalization allows us to relax the constraint on
our inner code $C_1$, i.e., $C_1$ is not necessary a perfect $q$-hash code at a cost of imposing a stronger constraint on the outer code.
That is, instead of requiring that $C_1$ is a perfect $q$-hash code, we only require that a fraction $|\mathcal{A}|/\binom{m}{q}$ of $q$-element sets of $C_1$ are separated.
\end{rmk}

If we choose $m=q$ in Lemma \ref{list-recovery-con}, then $|\mA|=1$. We obtain a random construction of perfect $q$-hash codes.
\begin{cor} Let $q\ge 2$.
Then there exists $q$-hash code of length $N$ and size at least $\left\lceil\frac{M}{3}\right\rceil$ as long as
\begin{equation}\label{eq:2}
{M\choose q}\left(1-\frac{q!}{q^q}\right)^N\leq \frac{M}{2q}.
\end{equation}
In particular, we have a random $q$-hash code with rate
\begin{equation}\label{eq:random}
R=\frac{\log_2 M}{N}=-\frac{1}{q-1}\log_2\left(1-\frac{q!}{q^q}\right)+\frac{O(1)}{N}.
\end{equation}
Hence, we have a probabilistic lower bound
\begin{equation}\label{eq:3} R_q\ge \frac{1}{q-1}\log_2\left(\frac1{1-{q!}/{q^q}}\right).\end{equation}
\end{cor}
\begin{proof} As ${M\choose q}\le\frac{M^q}{q!}$, the inequality
\begin{equation}\label{eq:y1} \frac{M^q}{q!}\left(1-\frac{q!}{q^q}\right)^N\leq \frac{M}{2q}\end{equation}
implies the inequality \eqref{eq:2}.
 Choose $M$ to be the largest integer satisfying the inequality \eqref{eq:y1} and consider the limit $\lim_{N\rightarrow\infty}\frac{\log_2 M}N$.
The desired equality \eqref{eq:random} follows.
\end{proof}
\subsection{A concatenation technique}
Let $C$ be a $q$-ary code of length $n$ and size $m$. Denote by $\mS(C)$ the collection of all $q$-element subsets of $C$ that are separated.

\begin{lemma}\label{lm:concatenation}
The following holds
\begin{equation}\label{eq:rate}
R_q\ge -\frac{1}{(q-1)n}\log_2\left(1-\frac{q!|\mS(C)|}{m^q}\right).
\end{equation}
\end{lemma}
\begin{proof}
Let $\pi$ be any bijection from $C$ to $[m]$.
Define $\mA:=\bigcup_{\{\bc_1,\ldots,\bc_q\}\in \mS(C)}\bigg\{\{\pi(\bc_1),\ldots,$ $\pi(\bc_q)\}\bigg\}$.
It is clear that $\mA\subseteq{[m]\choose q}$  and $|\mA|=|\mS(C)|$.
Lemma \ref{list-recovery-con} tells us that there exists an $m$-ary $\mA$-friendly code $C_1$ of length $n_1$ with rate
$$
R=-\frac{1}{(q-1)}\log_2\left(1-\frac{q!|\mA|}{m^q}\right)+\frac{O(1)}{n_1}.
$$
Let $C_2$ be the concatenation of $C_1$ with $C$, i.e.,
\[C_2:=\{\pi^{-1}(\bc)=(\pi^{-1}(c_1),\pi^{-1}(c_2),\dots,\pi^{-1}(c_{n_1})):\; \bc=(c_1,c_2,\dots,c_{n_1})\in C_1\}.\]
Clearly, the rate of $C_2$ is $R=-\frac{1}{n(q-1)}\log_2(1-\frac{q!|\mA|}{m^q})+\frac{O(1)}{n_1n_2}.$
It remains to show that $C_2$ is a perfect $q$-hash code.

Choose any $q$-element subset $\{\pi^{-1}(\bc_1),\pi^{-1}(\bc_2),\dots,\pi^{-1}(\bc_q)\}$ from $C_2$ with $\{\bc_1,\bc_2,\dots,\bc_q\}$ being a $q$-element subset of $C_1$. Since $C_1$ is $\mA$-friendly, there exists $i\in[N]$ such that $\{\proj_i(\bc_1), $ $\proj_i(\bc_2),\dots,\proj_i(\bc_q)\}\in \mA$.
This implies that $\{\pi^{-1}(\proj_i(\bc_1)),\ldots,\pi^{-1}(\proj_i(\bc_q))\}\in \mS(C)$ and thus $\{\pi^{-1}(\bc_1),\pi^{-1}(\bc_2),\dots,\pi^{-1}(\bc_q)\}$ is separated. The desired result follows from the definition of perfect $q$-hash codes.
\end{proof}

\begin{rmk}
Given a $q$-ary code $C$ of length $n$, Lemma \ref{lm:concatenation} tells us there must exist an outer code whose concatenation with $C$ yields a perfect $q$-hash code with rate  $-\frac{1}{n(q-1)}\log_2(1-\frac{q!|\mS(C_2)|}{m^q})$.
That means we only need to focus on finding good inner codes $C$ with large subset $\mS(C)$. In what follows, when we talk about concatenation, we only specify the inner code. The outer code is always given by Lemma \ref{lm:concatenation}.

\end{rmk}

\section{Lower bounds from MDS codes}

By Lemma \ref{lm:concatenation}, to have a good lower bound on $R_q$, one needs to find  a $q$-ary inner code $C$ of length $n$ such that $\mS(C)$ has large size for fixed $q$, $n$ and size $|C|$.
However, determining (or even estimating) the size of $\mS(C)$ for a given inner code $C$ with dimension at least $2$ seems very difficult. In this section, we estimate the size of $\mS(C)$ for some classes of codes and show that these inner codes give lower bounds on $R_q$ better than the probabilistic lower bound \eqref{eq:3}.

In this subsection, we investigate a promising candidate for the inner code, i.e., MDS code.
In general, the MDS code is defined over finite field. However, it is possible to define an MDS code over an abelian group as well.
The reason why we use abelian group instead of finite field is that we want our construction of a $q$-perfect hash code to exist for any
$q$ instead of merely prime power.
Let $G$ be an abelian group with $q$ elements and $G^n=G\times G\times\cdots\times G$.
Let $\bc=(c_1,\ldots,c_n)\in G^n$ and denote by $(\bc)_T=(c_i)_{i\in I}$ the codeword $\bc$ restricted to index set $T\subseteq [n]$. There are several equivalent definition for MDS codes. We use the following definition in our convenience.


\begin{defn}\label{defn:mds}
Let $G$ be an abelian group with $q$ elements. Let $C\subseteq G^n$ be a subset of size $q^k$.
Then, $C$ is a $[n,k]$-MDS code if and only if for any subset $T\subset [n]$ of size at most $k$ and any $\bx\in G^k$,
the set $\{\bc\in C:(\bc)_T=\bx\}$ is of size $q^{k-|T|}$.
\end{defn}
For each $i\in[n]$, define the set
\begin{equation}\label{eq:x12}\mA_i=\{\{\bc_1,\bc_2,\dots,\bc_q\}\subseteq C:\; \{\proj_i(\bc_1),\ldots,\proj_i(\bc_q)\}=G\}.\end{equation}
Thus, we have $\mS(C)=\cup_{i=1}^n\mA_i$.
For any subset $T$ of $[n]$, we denote by $\mA_T$ the set $\cap_{i\in T}\mA_i$. Let $B_i$ denote the number \begin{equation}\label{eq:x1} B_i= \sum_{T\subseteq [n], |T|=i} |\mA_T|.\end{equation}

\begin{lemma}\label{lm:inclusion-exclusion}  Let $C$ be a $[n,k]$-MDS code over abelian group $G$. Then
\begin{equation}\label{eq:8}
|\mS(C)|= \sum_{i=1}^{k}(-1)^{i-1}{n\choose i}q^{q(k-i)}(q!)^{i-1}+\sum_{i=k}^n(-1)^{i-1}B_i.
\end{equation}
\end{lemma}
\begin{proof}
First we claim that for any $j\leq k$ and subset $J\subseteq [n]$ with $|J|=j$, we have $|\mA_J|=q^{q(k-j)}(q!)^{j-1}$.
Note that if $\{\bc_1,\ldots,\bc_q\}\in \mA_J$, then $\{\proj_i(\bc_1),\ldots,\proj_i(\bc_q)\}=G$ for any $i\in J$.
This means the matrix
$$
M=\left(\begin{matrix}(\bc_1)_J\\ (\bc_2)_J\\ \vdots\\ (\bc_q)_J\end{matrix}\right)
$$
satisfies that each column of $M$ is a permutation of all elements in $G$. There are $(q!)^{j}$ such matrix $M$. Let us fix $M$ and denote by $\by_1,\ldots,\by_q$ the $q$ rows of $M$. Since $C$ is a MDS code of size $q^k$, Definition \ref{defn:mds} says that there are
$q^{k-j}$ codewords $\bc_i$ in $C$ with $(\bc_i)_J=\by_i$. This gives $(q!)^jq^{q(k-j)}$ different $q$-tuples $(\bc_1,\bc_2,\dots,\bc_q)$ with $\{\bc_1,\bc_2,\dots,\bc_q\}\in \mA_J$. It follows that the number of $q$-element sets in $\mA_J$ is  $(q!)^{j-1}q^{q(k-j)}$.

By the inclusion-exclusion principle, we have
\[|\mS(C)|=|\bigcup_{i=1}^{n}\mA_i|= \sum_{i=1}^{k}(-1)^{i-1}{n\choose i}q^{q(k-i)}(q!)^{i-1}+\sum_{i=k+1}^n(-1)^{i-1}B_i.\]
This completes the proof.
\end{proof}

By the equality \eqref{eq:8}, we have
\begin{eqnarray*}
|\mS(C)|&=& \sum_{i=1}^n(-1)^{i-1}{n\choose i}q^{q(k-i)}(q!)^{i-1}-\sum_{i=k+1}^n(-1)^{i-1}{n\choose i}q^{q(k-i)}(q!)^{i-1}+\sum_{i=k+1}^n(-1)^{i-1}B_i\\
&=&\frac{-q^{qk}}{q!q^{qn}}\left(-q^{qn}+(q^q-q!)^n\right)-\sum_{i=k+1}^n(-1)^{i-1}{n\choose i}q^{q(k-i)}(q!)^{i-1}+\sum_{i=k+1}^n(-1)^{i-1}B_i\\
&=&\frac{q^{qk}}{q!}\left(1-\left(1-\frac{q!}{q^q}\right)^n\right)-\sum_{i=k+1}^n(-1)^{i-1}{n\choose i}q^{q(k-i)}(q!)^{i-1}+\sum_{i=k+1}^n(-1)^{i-1}B_i.
\end{eqnarray*}
Thus, we have
\[1-\frac{q!|\mS(C)|}{q^{qk}}= \left(1-\frac{q!}{q^q}\right)^n+\sum_{i=k+1}^n(-1)^{i-1}{n\choose i}\left(\frac{q!}{q^q}\right)^{i}-\frac{q!}{q^{qk}}\sum_{i=k+1}^n(-1)^{i-1}B_i.\]
Hence, in order to beat  the probabilistic lower bound, we need to verify the following inequality for an $[n,k]$-MDS inner code $C$,
\begin{equation}\label{eq:10}
\sum_{i=k+1}^n(-1)^{i-1}{n\choose i}\left(\frac{q!}{q^q}\right)^{i}<\frac{q!}{q^{qk}}\sum_{i=k+1}^n(-1)^{i-1}B_i
\end{equation}

Lemma \ref{lm:inclusion-exclusion} shows that computing $|\mS(C)|$ is reduced to computing $B_i$ for $i=k+1,\dots, n$. However, if $k+1$ is too far from $n$, we have to compute many $B_i$ and this is rather difficult. The simplest case is $k=n-1$ where we need to compute only $A_n$. In this case, we use $[n,n-1]$-MDS code. Another gain for this choice is that the $[n,n-1]$-MDS code exists over any abelian group.
\begin{cor}\label{lm:code}
Let $q\ge 2$ be an integer and $G$ be an abelian group of order $q$. Define the $q$-ary MDS code $C=\{(x_1,\ldots,x_{n-1},\sum_{i=1}^{n-1}x_i):\; x_1,\ldots,x_{n-1}\in G\}.$ Let $A_n$ denote the cardinality of the set
\[\{\{\bc_1,\bc_2,\dots,\bc_q\}\subseteq C:\; \{\proj_i(\bc_1)),\ldots,(\proj_i(\bc_q))\}=G \ \mbox{for any $i\in[n]$}\}.\] Then
$|\mS(C)|= \frac{q^{q(n-1)}}{q!}\left(1-\left(1-\frac{q!}{q^q}\right)^n\right)-(-1)^{n-1}q^{-q}(q!)^{n-1}+(-1)^{n-1}A_n$.
\end{cor}
\begin{proof}
As $C$ is a $[n,n-1]$-MDS code, set $k=n-1$ in Lemma \ref{lm:inclusion-exclusion}.
\end{proof}
Combining \eqref{eq:10} and Corollary \ref{lm:code}, we obtain the following corollary.
\begin{cor}\label{cor:beat}
Let $q\ge 2$ be an integer and $A_n$ be the number given in Corollary \ref{lm:code}. If
\begin{equation}
(-1)^{n-1} A_n> (-1)^{n-1}\frac{(q!)^{n-1}}{q^q},
\end{equation}
Then there exist  families of perfect $q$-hash codes with rate better than the probabilistic lower bound \eqref{eq:3}.
\end{cor}

If $C$ is the code of length $3$ over $\ZZ_q$ in Corollary \ref{lm:code}, i.e, $C=\{(x,y,x+y):\; x,y\in\ZZ_q\}$, then determining $A_3$ given in  Corollary \ref{lm:code} is actually reduced to  the following well-known combinatorial problem: determining the number $s_q$ of pairs $(\pi_1,\pi_2)$ of bijections $[q]\rightarrow\ZZ_q$ such that $\pi_1+\pi_2$ is a bijection as well. The relation between $A_3$ and $s_q$ is $A_3=\frac{s_q}{q!}$.

The number $s_q$ has been studied somewhat extensively, but under a different guise \cite{B15,CGKN99,C00,V91,K07}. It is in general very difficult to determine the exact value of $s_q$ unless $q$ is an even number for which $s_q=0$.
To beat the probabilistic lower bound on $R_q$, we want to show $s_q> (\frac{q!^2}{q^q})$. That means,  we are only interested in the lower bounds on $s_q$. A generic lower bound is $s_q\ge 3.246^q\times q!$ for all odd $q$. However, there is still a very big gap between this lower bound and the aforementioned conjecture. For sufficiently large $q$, we actually has some asymptotically tight lower bound for $s_q$. We defer this discussion to the next subsection.
On the other hand, there are various algorithms to numerically approximate $s_q$ \cite{K07}. Precisely speaking, for many odd $q$ in the interval $[27,155]$, it is possible to approximate $s_q$ with certain accuracy. One can verify from these estimation that the probabilistic lower bound \eqref{eq:3} is improved for all odd integers $q$ in \cite{K07}.

By taking exact value of $s_q$ for all odd $q$ between $3$ and $25$ from \cite{K07}, we obtain the following result.
\begin{cor}\label{cor:x0}
There exists a family of perfect $q$-hash codes over $\ZZ_q$ with rate better than the probabilistic lower bound \eqref{eq:3} for all odd $q$ between $3$ and $25$. 
\end{cor}
\begin{proof}
By Corollary \ref{cor:beat}, it is sufficient to verify the inequality
\begin{equation}\label{eq:x0} \frac{s_q}{q!}> \frac{(q!)^2}{q^q}\end{equation}
 for all odd $q$ between $3$ and $25$. Taking the values of $s_q$ from Table I of \cite{K07} gives the desired claim.
\end{proof}

\begin{rmk}
From Table \ref{fig:table1}, we observe that the ratio $A_3$ over
$\frac{(q!)^2}{q^q}$ grows slowly but monotonically. In fact, we will see that this ratio is asymptotically equal to $\frac{q}{\sqrt{e}}$ in the next section.

\begin{table}

\begin{center}

\begin{tabular}{||c |c|c|c|c|c|c||}
  \hline
  $\ZZ_q$  & $\ZZ_5$  & $\ZZ_7$  & $\ZZ_9$ &  $\ZZ_{11}$& $\ZZ_{13}$ &  $\ZZ_{15}$ \\ \hline
  $A_3$ & $15$ & $133$ & $2025$ & $37851$ & $1.03\times 10^6$ &  $3.63\times 10^7$   \\\hline
  $\frac{(q!)^2}{q^q}$ & $4.6$ & $30.8$ & $339.9$ & $5584.6$ & $1.28\times 10^5$ & $3.90\times 10^6$ \\ \hline
  Ratio & $3.26$ & $4.32$ & $5.96$ & $6.78$ & $8.04$ & $9.30$ \\ \hline\hline

  $\ZZ_q$ &$\ZZ_{17}$ & $\ZZ_{19}$ & $\ZZ_{21}$& $\ZZ_{23}$ & $\ZZ_{25}$ &\\ \hline
  $A_3$ & $1.60\times 10^9$ & $8.76\times 10^{10}$ & $5.77\times 10^{12}$  &  $4.52\times 10^{14}$ & $4.16\times 10^{16}$ &\\ \hline
  $\frac{(q!)^2}{q^q}$ & $1.52*10^8$ & $7.47\times 10^{9}$ & $4.47\times 10^{11}$ & $3.2\times 10^{13}$ & $2.70\times 10^{15}$ & \\ \hline
  Ratio & $10.53$ & $11.71$ & $12.93$ & $14.12$ & $15.4$ & \\ \hline

\end{tabular}
\end{center}
\caption{The comparison between $A_3$ and $\frac{(q!)^2}{q^q}$ for small odd $q$.}
\label{fig:table1}
 \end{table}
\end{rmk}


For even $q$, we have $s_q=0$. We turn to other abelian groups instead of $\ZZ_q$.
\begin{cor}\label{cor:xx0}
There exists a family of perfect $q$-hash code with rate better than the probabilistic lower bound \eqref{eq:3} for $q=4,8,9,12$.
\end{cor}
\begin{proof}
Let $C$ be a code with the form
$$
C=\{(x,y,x+y):x,y\in \F_q\}.
$$
for $q=4,8,9$.
When $q=12$, we let $C=\{(x,y,x+y):x,y\in \F_3\times\F_4\}$.
With the help of computer search, we present the values $A_3$ of $C$ in Table \ref{table:sm}.
\begin{table}[h]

\begin{center}

\begin{tabular}{||c |c|c|c|c|c|c||}
  \hline
  $q$  & $\F_4$  & $\F_8$  & $\F_9$ &  $\F_{3}\times \F_4$ \\ \hline
  $A_3$ & $8$ & $384$ & $2241$ & $198144$   \\\hline
   $\frac{(q!)^2}{q^q}$ & $2.25$ & $96.89$ & $339.9$ & $25733.5$  \\ \hline\hline
\end{tabular}
\end{center}
\caption{The comparison between $A_3$ and $\frac{(q!)^2}{q^q}$ for $q$.}\label{table:sm}

 \end{table}
\end{proof}

\begin{rmk}
The lower bound on $R_3$ given in \cite{KM} is $R_3\ge \frac 14\log_2\frac95$.
Let $C$ be a ternary $[4,2]$-MDS code. The computer search shows that $|\mS(C)|=84$. By  Lemma \ref{lm:concatenation}, we also obtain the same lower bound $R_3\ge \frac 14\log_2\frac95$.
\end{rmk}

This remark indicates that $q$-ary MDS codes of larger length may lead to a better lower bound on $R_q$ than $q$-ary [3,2]-MDS codes. This is further confirmed by the following example for $q=4$.

\begin{cor}
There exists a family of perfect $4$-hash code over $\F_4$ with rate at least $0.049586$. This is better than both the lower bound given in Corollary \ref{cor:xx0} and  the probabilistic lower bound.
\end{cor}
\begin{proof}
Assume $\F_4=\{0,1,\alpha,\alpha+1\}$.
Consider a $[5,2]$-MDS code:
$$
C=\{(a,b,a+b,a\alpha+b,a(\alpha+1)+b): a,b\in \F_4\}.
$$
By computer search, we find that
there are $1100$ out of $\binom{32}{4}$ $4$-element subsets of $C$ that are separated. Plugging it
parameters into Lemma \ref{lm:concatenation}, we obtain perfect $4$-hash code with rate $0.049586$.
\end{proof}

\section{A lower bound for big $q\neq 2\pmod 4$}

We need a lower bound on $s_q$. For big $q$ we can use a rather precise asymptotic estimate proved in \cite{BMM}.
 Their result settles a conjecture saying that, for all odd $q$, the number $s_q$ lies in between $c_1^n n!^2$ and $c_2^n n!^2$ for some constants $c_1,c_2$. This conjecture is recently confirmed in \cite{BMM}. They even close the gap by showing $c_1=c_2=\frac{1}{e}+o(1)$.

\begin{prop}[\cite{BMM}]\label{prop:key}
Let $q$ be an odd integer. Then, the number $s_q$ is  $(\frac{1}{\sqrt{e}}+o(1))\frac{q!^3}{q^{q-1}}$, and hence $A_3$ defined in Corollary \ref{lm:code}  is $(\frac{1}{\sqrt{e}}+o(1))\frac{q!^2}{q^{q-1}}$.
\end{prop}

Plugging $A_3$ in Proposition \ref{prop:key} into \eqref{eq:8} and \eqref{eq:rate} gives the following theorem.

\begin{theorem}\label{thm:main}
For every odd integer $q$, one has
$$
R_q\ge-\frac{1}{3(q-1)}\log_2\bigg(1-3\frac{q!}{q^q}+3\frac{(q!)^2}{q^{2q}}-\left(\frac{1}{\sqrt{e}}+o(1)\right)\frac{(q!)^3}{q^{3q-1}}\bigg).
$$
Moreover, for every sufficiently large odd $q$ this rate is bigger than that given by the probabilistic lower bound.
\end{theorem}
\begin{proof}
From \eqref{eq:random}, it suffices to show $A_3>\frac{(q!)^2}{q^q}.$
For large odd $q$, this inequality is reduced to prove
$\left(\frac{1}{\sqrt{e}}+o(1)\right)\frac{(q!)^3}{q^{3q-1}}>\frac{(q!)^3}{q^{3q}}.$
This holds as $\frac{1}{\sqrt{e}}+o(1)>\frac{1}{q}$ for sufficiently large $q$.
\end{proof}

As $s_q=0$ for even $q$, we have to replace group $\ZZ_q$ by other abelian groups of order $q$.
Recently, Eberhard \cite{Ebe2017} extended Proposition \ref{prop:key} to any abelian group $G$ with
$\sum_{x\in G}x=0$ and size $q$. In fact, he proved an even more general result.

\begin{prop}[\cite{Ebe2017}]\label{prop:group}
Let $G$ be an abelian group of size $q$ and $f$ is a function from $[q]$ to $G$ such that $\sum_{i=1}^q f(i)=\sum_{x\in F}x$.
Let $S$ be the collection of bijections that maps $[q]$ to $G$. Then, the set of $\{(\pi_1,\pi_2,\pi_3)\in S^3:\pi_1(i)+\pi_2(i)+\pi_3(i)=f(i), \forall i\in [q]\}$ is of size $(\frac{1}{\sqrt{e}}+o(1))\frac{q!^3}{q^{q-1}}$.
\end{prop}
Let $G$ be an ablian group of size $q=0 \pmod 4$ and $f$ be a zero function, i.e., $f(i)=0$ for all $i\in G$. We have the following corollary.
\begin{cor}\label{cor:group}
Let $s_G$ be the number of pairs $(\pi_1,\pi_2)$ of bijections $[q]\rightarrow G$ such that $\pi_1+\pi_2$ is a bijection as well.
Then, $s_G$ is $(\frac{1}{\sqrt{e}}+o(1))\frac{q!^3}{q^{q-1}}$.
\end{cor}

\begin{theorem}\label{thm:3.11}
For every integer $q$ with $q=0 \pmod 4$, one has
$$
R=-\frac{1}{3(q-1)}\log_2\bigg(1-3\frac{q!}{q^q}+3\frac{(q!)^2}{q^{2q}}-\left(\frac{1}{\sqrt{e}}+o(1)\right)\frac{(q!)^3}{q^{3q-1}}\bigg).
$$
Moreover, for every sufficiently large  $q$, this rate is bigger than that given by the probabilistic lower bound.
\end{theorem}
\begin{proof}
Since $q=0 \pmod 4$, let $q=2^rp$ with an odd integer $p$ and $r\geq 2$. Let $G=\F_{2^r}\times \ZZ_p$. It is clear that $G$ is an abelian group and $\sum_{x\in G}x=0$. Define the code $C:=\{(x,y,x+y): x,y\in G\}$. Then, $C$ is an MDS code with dimension $2$ and length $3$. It remains to bound $A_3$. This is equivalent to counting the pair of bijections $(\pi_1,\pi_2)$: $[q]\rightarrow F$ such that $\pi_1+\pi_2$ is a bijection as well.
Corollary \ref{cor:group} says that the number $A_3$ of $C$ is $\frac{s_G}{q!}=(\frac{1}{\sqrt{e}}+o(1))\frac{q!^2}{q^{q-1}}$. Plugging $A_3$ into \eqref{eq:8} and \eqref{eq:rate} gives the desired result.
\end{proof}

The lower bounds given in Theorems \ref{thm:main} and \ref{thm:3.11} make use of linear codes over an abelian group of length $3$ and dimension $2$. As we have seen, this code does not always give the best lower bound. In the rest of this section, we show that [4,2]-MDS code over an abelian group provides a better lower bound than those given in  Theorems \ref{thm:main} and \ref{thm:3.11}.

\begin{lemma}\label{lem:x10}  Let $q\ge 3$ be an odd integer. Consider the code
\[C=\{(x,y,x+y,x-y):\; x,y\in \ZZ_q\}.\]
Then one has
\[|\mS(C)|\ge {4\choose 1}q^q-{4\choose 2}q!+3\frac{s_q}{q!}.\]
\end{lemma}
\begin{proof} Similar to the arguments in Lemma \ref{lm:inclusion-exclusion}, we have
\[|\mS(C)|={4\choose 1}q^q-{4\choose 2}q!+A_3-A_4,\]
where $B_i$ is the number defined in \eqref{eq:x1}. For any subset $T\subseteq [4]$ of size $3$, we claim that $|\mA_T|=\frac{s_q}{q!}$. To prove this claim, let us only consider  the case where $T=\{1,3,4\}$. Note that $C$ can be rewritten as $C=(2^{-1}(w+z),2^{-1}(w-z),w,z):\; w,z\in \ZZ_q\}$. If the third and fourth positions of $\ZZ_q$ are associated  with two permutations $\pi_1$ and $\pi_2$, respectively, then the first position forms a permutation of $\ZZ_q$ if and only if $2^{-1}(\pi_1+\pi_2)$ is a permutation of $\ZZ_q$. This is equivalent to that $\pi_1+\pi_2$ is a permutation of $\ZZ_q$. Hence, we have $|\mA_T|=\frac{s_q}{q!}$. We can similarly prove the claim for other three cases.

Hence, we have $A_3=4\frac{s_q}{q!}$. As we have $A_4\le |\mA_{[3]}|=\frac{s_q}{q!}$, the desired result follows.
\end{proof}

\begin{theorem}\label{thm:3.13}
For any odd integer $q\ge 3$,  one has
$$
R_q\ge-\frac{1}{4(q-1)}\log_2\bigg(\left(1-\frac{q!}{q^q}\right)^4-\left(\frac{3q}{\sqrt{e}}+o(q)\right)\left(\frac{q!}{q^q}\right)^3\bigg).
$$
Moreover, for every sufficiently large  odd $q$, this rate is bigger than that given in Theorem {\rm \ref{thm:main}}.
\end{theorem}
\begin{proof} Let $C$ be the $q$-ary code given in Lemma \ref{lem:x10}. Then we have
\begin{equation}
1-\frac{q!|\mS(C)|}{|C|^q}\le 1-{4\choose 1}\frac{q!}{q^q}+{4\choose 2}\left(\frac{q!}{q^q}\right)^2-3\frac{s_q}{q^{2q}}=\left(1-\frac{q!}{q^q}\right)^4-\left(\frac{3q}{\sqrt{e}}+o(q)\right)\left(\frac{q!}{q^q}\right)^3.
\end{equation}
The first claim is proved. To prove the second claim, it is sufficient to show that
\[
\left(\left(1-\frac{q!}{q^q}\right)^4-\left(\frac{3q}{\sqrt{e}}+o(q)\right)\left(\frac{q!}{q^q}\right)^3\right)^{1/4}<\bigg(1-3\frac{q!}{q^q}+3\frac{(q!)^2}{q^{2q}}-\left(\frac{1}{\sqrt{e}}+o(1)\right)\frac{(q!)^3}{q^{3q-1}}\bigg)^{1/3},
\]
i.e.,
\begin{equation}\label{eq:y10}
\left(\left(1-\frac{q!}{q^q}\right)^4-\left(\frac{3q}{\sqrt{e}}+o(q)\right)\left(\frac{q!}{q^q}\right)^3\right)^3<\bigg(\left(1-\frac{q!}{q^q}\right)^3-\left(\frac{q}{\sqrt{e}}+o(q)\right)\left(\frac{q!}{q^q}\right)^3\bigg)^4.
\end{equation}
The left-hand side of \eqref{eq:y10} is
\begin{equation}\label{eq:y12}
\left(1-\frac{q!}{q^q}\right)^{12}-\left(\frac{9q}{\sqrt{e}}+o(q)\right)\left(\frac{q!}{q^q}\right)^3(1+o(1))=\left(1-\frac{q!}{q^q}\right)^{12}-\frac{9q}{\sqrt{e}}\left(\frac{q!}{q^q}\right)^3(1+o(1)).
\end{equation}
Similarly, the right-hand side of \eqref{eq:y10} is
\begin{equation}\label{eq:y11}
\left(1-\frac{q!}{q^q}\right)^{12}-\frac{4q}{\sqrt{e}}\left(\frac{q!}{q^q}\right)^3(1+o(1)).
\end{equation}
As the number of \eqref{eq:y12} is less than the number of \eqref{eq:y11}, the second claim follows.
\end{proof}
Similar to the case where $q$ is odd, we can also improve the lower bound given in Theorem \ref{thm:3.11} if $q$ is divisible by $4$.

\begin{theorem}\label{thm:3.14}
For any integer $q$ with $q=0\pmod{4}$,  one has
$$
R_q\ge-\frac{1}{4(q-1)}\log_2\bigg(\left(1-\frac{q!}{q^q}\right)^4-\left(\frac{3q}{\sqrt{e}}+o(q)\right)\left(\frac{q!}{q^q}\right)^3\bigg).
$$
Moreover, for every sufficiently large  $q$, this rate is bigger than that given in Theorem {\rm \ref{thm:3.11}}.
\end{theorem}
\begin{proof} {\it Case 1:} $q=2^r$ for some integer $r\ge 2$. Choose an element $\Ga\in \F_q-\F_2$ and consider the code $C=\{(x,y,x+y,x+\Ga y):\; x,y\in \F_q\}.$ Then as in the proof of Lemma \ref{lem:x10}, one can show that
$|\mS(C)|\ge {4\choose 1}q^q-{4\choose 2}q!+3\frac{s_q}{q!}.$ By the same arguments in the proof of Theorem \ref{thm:3.13}, we obtain the desired result.

{\it Case 2:} $q=2^rp$ for some integer $r\ge 2$ and an odd $p\ge 3$. Choose an element $\Ga\in \F_{2^r}-\F_2$ and consider the ring $\F_{2^r}\times\ZZ_p$. Define the code $C=\{(x,y,x+y,x+(\Ga,-1) y):\; x,y\in \F_{2^r}\times\ZZ_p\}.$
$C$ is a $[4,2]$-MDS code by observing that both $(\Ga,-1)$ and $(\Ga,-1)-(1,1)=(\Ga-1,-2)$ are invertible elements in $\F_{2^r}\times\ZZ_p$.
The desired result then follows from the similar arguments in the proofs of Lemma \ref{lem:x10} and Theorem \ref{thm:3.13}.
\end{proof}

\section{A  lower bound for $q=2 \pmod{4}$}
The previous section provides a construction of a $q$-perfect hash code for any $q\neq 2 \mod 4$. This construction does not work for the case $q=2 \pmod 4$ because if $\pi_1$ and $\pi_2$ are two bijections from $\ZZ_q$ with $q=2\pmod{4}$, the sum $\pi_1+\pi_2$ is not a bijection. To see this, any bijection $\pi$ satisfies that
$$
\sum_{i\in \ZZ_q}\pi(i)=\sum_{i\in \ZZ_q}i=(q-1)\times \frac{q}{2} \bmod q
$$
which is not divisible by $q$ when $q=2\pmod{4}$. However, the sum of two bijections satisfies that
$$
\sum_{i\in \ZZ_q}\bigg(\pi_1(i)+\pi_2(i)\bigg)=q(q-1)=0 \bmod q.
$$
It is clear that $s_q$ is $0$ in this case. Therefore, we have to look for other tools to achieve our goal.

In the rest of this section, we assume that $q=2\pmod 4$. Let $C=\{(x,y,-(x+y)),x,y\in \ZZ_q\}\cup \{(x,y,-(x+y)+\frac{q}{2}),x,y\in \ZZ_q\}$. It is clear that $C$ is the union of two MDS codes with $C_1=\{(x,y,-(x+y)),x,y\in \ZZ_q\}$ and $C_2=C_1+(0,0,\frac{q}{2})$. Recall

$$
\mA_i=\{\{\bc_1,\bc_2,\ldots,\bc_q\}\subseteq C:\; \{\proj_i(\bc_1),\ldots,\proj_i(\bc_q)\}=\ZZ_q\}.
$$
We want to estimate the size of $S(C)=|\mA_1\cup\mA_2\cup\mA_3|$.

\begin{lemma}\label{lm:set}
Let $C$ be the code and $\mA_i$ be the set defined above. Then, we have
\begin{equation}\label{eq:SC}
|\mS(C)|=3(2q)^q-3\times2^q(q!)+|\mA_1\cap\mA_2\cap\mA_3|.
\end{equation}
\end{lemma}
\begin{proof}

By the inclusion-exclusion principle, we have
$$
|\mS(C)|=\sum_{i=1}^3 |\mA_i|-(|\mA_1\cap \mA_2|+|\mA_2\cap \mA_3|+|\mA_1\cap\mA_3|)+|\mA_1\cap \mA_2\cap \mA_3|.
$$
The first two terms can be calculated precisely. Due to the symmetry and MDS property, it suffices to calculate $|\mA_1|$ and $|\mA_1\cup\mA_2|$.
Note that $C$ is the union of two MDS codes $C_1$ and $C_2$. This means, given any bijection $\pi=(x_1,\ldots,x_q)$ from $[q]$ to $\ZZ_q$, there are
$2q$ codewords $\bc_i$ in $C$ such that $\proj_1(\bc_i)=x_i$ for any $i\in [q]$. Note that $x_1,\ldots,x_q$ are all distinct, thus the number of tuples $(\bc_1,\bc_2,\ldots,\bc_q)\in C^q$ such that $(\proj_1(\bc_1),\ldots,\proj_1(\bc_q))=\pi$
is $(2q)^{q}$. Since there are $q!$ bijections, we conclude that
$$
\sum_{i=1}^3 |\mA_i|=3|\mA_1|=3\times \frac{(2q)^q(q!)}{q!}=3\times (2q)^q.
$$
We proceed to calculate $|\mA_1\cap \mA_2|$. Let $\pi_1=(x_1,\ldots,x_q)$ and $\pi_2=(y_1,\ldots,y_q)$ be any bijections from $[q]$ to $\ZZ_q$. Since $C_1$ and $C_2$ are $[3,2]$-MDS codes, there are exactly two codewords $\bc_i$, one from $C_1$ and another one from $C_2$ such that $(\proj_1(\bc_i),\proj_2(\bc_i))=(x_i,y_i)$ for all $i\in [q]$.
Since there are $(q!)^2$ pairs of bijections $(\pi_1,\pi_2)$, we conclude that
$$
|\mA_1\cap \mA_2|+|\mA_2\cap \mA_3|+|\mA_1\cap\mA_3|=3|\mA_1\cap \mA_2|=3\times \frac{2^q(q!)^2}{q!}=3\times 2^q q!
$$
The proof is completed.
\end{proof}
Plugging the Equation \eqref{eq:SC} into the Equation \eqref{eq:rate} gives
\begin{eqnarray}\label{eq:Rq}
R_q&\ge& -\frac{1}{(q-1)3}\log_2\left(1-\frac{q!|\mS(C)|}{(2q^2)^q}\right)\\ \nonumber
&=&-\frac{1}{(q-1)3}\log_2\left(1-3\times\frac{q!}{q^q}+3\times (\frac{q!}{q^q})^2-\frac{q! |\mA_1\cap\mA_2\cap\mA_3|}{2^qq^{2q}}\right).
\end{eqnarray}

To get a good lower bound on $R_q$, we have to  find a  reasonable lower bound on the size of $\mA_1\cap\mA_2\cap\mA_3$.
\begin{theorem}\label{thm:newbound}
There exists a $q$-perfect hash code with rate at least
$$
R_q\ge-\frac{1}{3(q-1)}\log_2\left(1-3\times\frac{q!}{q^q}+3\times (\frac{q!}{q^q})^2-(\frac{1}{2\sqrt{e}}+o(1))\frac{(q!)^3}{q^{3q-1}}\right)
$$
for $q=2\pmod 4$.
Moreover, for every sufficiently large  $q$, this rate is bigger than that given in Theorem {\rm \ref{thm:3.11}}.
\end{theorem}
\begin{proof}
We choose $C=C_1\cup C_2$ as the inner code and the outer code is defined by Lemma \ref{lm:concatenation} accordingly.
Thanks to Lemma \ref{lm:set}, it remains to lower bound the size of $\mA_1\cap\mA_2\cap\mA_3$. Let
$\{\bc_1,\ldots,\bc_q\}$ be any set belonging to $\mA_1\cap\mA_2\cap\mA_3$ with $\bc_i=(x_i,y_i,z_i)$. Let
$\pi_1:=(x_1,\ldots,x_q),\pi_2:=(y_1,\ldots,y_q),\pi_3=(z_1,\ldots,z_q)$ are three bijections by the definition of $\mA_1\cap\mA_2\cap\mA_3$. Assume that there are $\ell$ codewords of $\{\bc_1,\ldots,\bc_q\}$ from $C_1$ and $q-\ell$ from $C_2$.
Without loss of generality, let $\bc_1,\bc_2,\ldots,\bc_\ell\in C_1$ and $\bc_{\ell+1},\ldots,\bc_{q}\in C_2$. By the definition of $C_1$ and $C_2$, we have $\pi_1(i)+\pi_2(i)+\pi_3(i)=0$ for $i=1,\ldots,\ell$ and $\pi_1(i)+\pi_2(i)+\pi_3(i)=\frac{q}{2}$ for $i=\ell+1,\ldots,q$. Let $f$ be a map from $[q]$ to $\ZZ_q$ such that $f(i)=0$ for $i=1,\ldots,\ell$ and $f(i)=\frac{q}{2}$ for $i=\ell+1,\ldots,q$. If $\ell$ is odd number, then $\sum_{i\in [q]}f(i)=\frac{q}{2}=\sum_{x\in \ZZ_q}x$.
By Proposition \ref{prop:group}, when $\ell$ is odd number, the number of triples of bijections $(\pi_1,\pi_2,\pi_3)$ with $\pi_1+\pi_2+\pi_3=f$ is $(\frac{1}{\sqrt{e}}+o(1))\frac{q!^3}{q^{q-1}}$. Since there are $\binom{q}{\ell}$ ways to choose a $\ell$-codewords subset from $C_1$, the number of codewords $(\bc_1,\ldots,\bc_q)$ belonging to $\mA_1\cap\mA_2\cap\mA_3$ is at least
$$
\sum_{i=0}^{\frac{q-2}{2}}\binom{q}{2i+1}(\frac{1}{\sqrt{e}}+o(1))\frac{(q!)^3}{q^{q-1}q!}=(\frac{1}{\sqrt{e}}+o(1))\frac{2^{q-1}(q!)^2}{q^{q-1}}
$$
Plug this value into Equation \eqref{eq:Rq} yields the desired result.
\end{proof}
\begin{rmk}
The probabilistic lower bound \eqref{eq:3} can be written as
$$
-\frac{1}{3(q-1)}\log_2\left(1-3\times\frac{q!}{q^q}+3\times (\frac{q!}{q^q})^2-\frac{(q!)^3}{q^{3q}}\right).
$$
It is clear that the lower bound given by Theorem \ref{thm:newbound} is better as
$$
\frac{q}{2\sqrt{e}}\times\frac{(q!)^3}{q^{3q}}>\frac{(q!)^3}{q^{3q}}.
$$
\end{rmk}

We note that our construction can be applied for any $q=2\pmod 4$. For small $q$, we do the calculation with the help of the computer. Our numerical result shows that our construction beats the probabilistic lower bound for $q=6,10,14$. We believe that such trend should keep as well when $q$ grows. In conclusion, this construction is a very promising candidate to beat the probabilistic lower bound for all $q=2\pmod 4$. The following result summarizes our numerical computations for $q=6,10,14$.

\begin{theorem}\label{thm:smallq}
From our new construction, the following holds,
$R_6\geq 0.004488, R_{10}\geq 5.8180030\times 10^{-5}, R_{14}\geq 8.7066030151\times 10^{-7}$. In comparison, the previous probabilistic lower bound yields $R_6\geq 0.004487, R_{10}\geq 5.8180021\times 10^{-5}, R_{14}\geq 8.706603140\times 10^{-7}$.
\end{theorem}

\section{Lower bounds on $R_5$ and $R_7$}
In Section 3, we let inner code to be the MDS code $C$ and estimate the size $|\mS(C)|$ either numerically or asymptotically. However, MDS codes do not always provide the best lower bound on $R_q$. In this section, we present a class of nonlinear inner code $C$ where many $q$-element subsets are separated.
\begin{lemma}\label{thm:inner}
Assume $q$ is a prime. There exists a code $C$ over $\ZZ_q$ with length $q$ and size $2q$ such that
$|\mS(C)|=2^qq-2(q-1)$.
\end{lemma}
\begin{proof}
Let $C_1=\{\bc_1=(0,1\ldots,q-1),\bc_2=(1,2\ldots q-1,0),\ldots,\bc_q=(q-1,0,\ldots,q-2)\}$, i.e., $C_1$ consists of the codeword $(0,1\ldots,q-1)$ and its $i$th shifts for $i=1,\ldots,q-1$. Let $C_2=\{i\cdot\bi:\; 0\le i\le q-1\}$, where $\bi$ stands for all-one vector of length $q$. Let $C=C_1\cup C_2$.
Obviously, $C$ has length $q$ and size $2q$. It remains to show that $|\mS(C)|=2^qq-2(q-1)$.

We pick any $0<i<q$ codewords $\bc_1,\ldots,\bc_i$ from $C_1$.
Denote by $\bc_j=(c_{j,1},\ldots,c_{j,q})$ for $j\in [q]$. For $t\in [q]$, let
$B_t:=\{c_{1,t},c_{2,t},\ldots,c_{i,t}\}$ be the collection of the $t$-th components of $\bc_1,\ldots,\bc_i$.
It is clear that $|B_t|=i$ by observing that all codewords in $C_1$ have distinct values on each coordinate.
Moreover, we can show that $B_1,\ldots,B_q$ are distinct if $0<i<q$. Assume not and we have $B_1=B_a$ for some $a\in [q]$.
The structure of code $C_1$ tells us that $c_{j,a}=c_{j,1}+a-1$ for $j=1,\ldots,i$. This coupled with $B_1=B_a$ implies that both $c_{1,1}$ and $c_{1,1}+a-1$ belong to $B_1$. Continue this argument and we finally arrive at
$\{c_{1,1},c_{1,1}+a-1,\ldots,c_{1,1}+(q-1)(a-1)\}\subseteq B_1$.
It is clear that $c_{1,1},c_{1,1}+a-1,\ldots,c_{1,1}+(q-1)(a-1)$
are distinct which contradicts our assumption that $|B_t|=i<q$.

Now, we know that $B_1,\ldots,B_q$ are distinct. For each set $B_t=\{c_{1,t},c_{2,t},\ldots,c_{i,t}\}$,
we choose a $(q-i)$-element set $A_t:=\{\mathbf{i}: i\notin B_t\}\subseteq C_2$. It is clear that $\bc_1,\ldots,\bc_i$ and the codewords in $B$ have distinct symbols on $i$-th coordinate. Moreover, for each value $t$, the set $A_t$ is distinct due to the fact that $B_1,\ldots,B_q$ are distinct.
That means, for any $0<i<q$-element set of $C_1$, we could obtain $q$ distinct $q$-element sets of $C$ that are separated.
If $i=0$ or $i=q$, it is clear that the only $q$-element set that are separated is $C_1$ or $C_2$. Thus,
the total number of $q$-element sets of $C$ that are separated is
$\sum_{i=1}^{q-1}q\binom{q}{i}+2=2^qq-2(q-1)$.
\end{proof}

Combined this construction with Lemma \ref{lm:concatenation} gives following lower bounds on $R_q$ for $q=5$ and $7$.

\begin{cor}
One has $R_5\ge 0.01452$ and $R_7\ge 0.001483$. Furthermore, the lower bounds on $R_5$ and $R_7$ given in this corollary are better than those in Corollary \ref{cor:x0} and the probabilistic lower bound.
\end{cor}
\begin{proof} Take the inner code to be the code in Lemma \ref{thm:inner} for $q=5$ and $7$, respectively.
The desired result follows from  Lemma ~\ref{thm:inner} and  \ref{lm:concatenation}.
\end{proof}

Let us end this section by tabulating our best lower bound, denoted by $R_{new}$, obtained in this paper and the probabilistic lower bound denoted by $R_{ran}$ for some small $q$.
We omit cases for $q\geq 12$. 
\begin{table}[h]\label{fig:table2}

\begin{center}
\begin{tabular}{||c |c|c|c|c||}
  \hline
 $q$  & $4$ & $5$ & $6$ & $7$   \\ \hline
$R_{new}$ & $0.0495$ & $0.01452$ & $ 0.004488$ &  $0.001483$   \\ \hline
 $R_{ran}$  & $0.0473$ & $0.01412$ & $ 0.004477$ & $0.001476$  \\ \hline\hline
 $q$ &   $8$   & $9$ &  $10$ &  $11$  \\ \hline
 $R_{new}$ &  $4.95909\times 10^{-4}$  & $1.689931\times 10^{-4}$ & $5.8180030\times 10^{-5}$ & $2.01855746\times 10^{-5}$   \\ \hline
  $R_{ran}$ &   $4.95905\times 10^{-4}$ & $1.689929\times 10^{-4}$&  $5.8180021\times 10^{-5}$ & $2.01855739\times 10^{-5}$   \\ \hline
\end{tabular}
\end{center}
\caption{New lower bounds versus the probabilistic lower bounds}
 \end{table}

\section*{Acknowledgements}
We are grateful to Venkat Guruswami who brought this topic to us. He gave a talk on his paper \cite{GR} in our seminar when he was visiting Nanyang Technological University in 2018.

\bibliographystyle{plain}
\bibliography{hashcode}

\begin{thebibliography}{10}

\bibitem{AYZ}
Noga Alon, Raphael Yuster, and Uri Zwick.
\newblock {Color-Coding}.
\newblock {\em J. {ACM}}, 42(4):844--856, 1995.

\bibitem{Arikan}
Erdal Arikan.
\newblock Upper bound on the zero-error list-coding capacity.
\newblock {\em Information Theory, IEEE Transactions on}, 40:1237 -- 1240,
  1994.

\bibitem{B15}
Miklos Bona.
\newblock {\em {Handbook of enumerative combinatorics}}.
\newblock Discrete Mathematics and Its Applications. CRC Press, Hoboken, NJ,
  2015.

\bibitem{C00}
C.~Cooper.
\newblock A lower bound for the number of good permutations.
\newblock {\em Nat. Acad. Sci. Ukraine}, 213:15--25, 2000.

\bibitem{CGKN99}
C.~Cooper, R.~Gilchrist, I.~N. Kovalenko, and D.~Novakovic.
\newblock Estimation of the number of ``good'' permutation with applications to
  cryptography.
\newblock {\em Cybernetics and Systems Analysis}, 35(5):688--693, Sep 1999.

\bibitem{CD20}
Simone Costa and Marco Dalai.
\newblock New bounds for perfect k-hashing.
\newblock {\em CoRR}, abs/2002.11025, 2020.

\bibitem{DGR17}
M.~{Dalai}, V.~Guruswami, and J.~{Radhakrishnan}.
\newblock An improved bound on the zero-error list-decoding capacity of the 4/3
  channel.
\newblock In {\em 2017 IEEE International Symposium on Information Theory
  (ISIT)}, pages 1658--1662, June 2017.

\bibitem{BMM}
S.~Eberhard, F.~Manners, and R.~Mrazovi\'{c}.
\newblock Additive triples of bijections, or the toroidal semiqueens problem.
\newblock {\em Journal of the European Mathematical Society}, 21(2):441--463,
  2019.

\bibitem{Ebe2017}
Sean Eberhard.
\newblock More on additive triples of bijections.
\newblock {\em CoRR}, abs/1704.02407, 2017.

\bibitem{E88}
P.~{Elias}.
\newblock Zero error capacity under list decoding.
\newblock {\em IEEE Transactions on Information Theory}, 34(5):1070--1074, Sep.
  1988.

\bibitem{FCD}
Stefano~Della Fiore, Simone Costa, and Marco Dalai.
\newblock Further strengthening of upper bounds for perfect k-hashing.
\newblock {\em CoRR}, abs/2012.00620, 2020.

\bibitem{FK}
M.~Fredman and J.~Koml\'{o}s.
\newblock {On the Size of Separating Systems and Families of Perfect Hash
  Functions}.
\newblock {\em SIAM Journal on Algebraic Discrete Methods}, 5(1):61--68, 1984.

\bibitem{GR}
Venkatesan Guruswami and Andrii Riazanov.
\newblock {Beating Fredman-Koml{\'o}s for Perfect k-Hashing}.
\newblock In {\em 46th International Colloquium on Automata, Languages, and
  Programming (ICALP 2019)}, volume 132, pages 92:1--92:14, Dagstuhl, Germany,
  2019.

\bibitem{TT}
Torben Hagerup and Torsten Tholey.
\newblock {Efficient Minimal Perfect Hashing in Nearly Minimal Space}.
\newblock In {\em STACS 2001}, pages 317--326, Berlin, Heidelberg, 2001.
  Springer Berlin Heidelberg.

\bibitem{K86}
J.~K{\"o}rner.
\newblock Fredman-koml\'{o}s bounds and information theory.
\newblock {\em SIAM Journal on Algebraic Discrete Methods}, pages 560--570,
  1986.

\bibitem{KM}
J.~K{\"o}rner and K.~Marton.
\newblock {New Bounds for Perfect Hashing via Information Theory}.
\newblock {\em European Journal of Combinatorics}, 9(6):523--530, 1988.

\bibitem{K07}
N.~Kuznetsov.
\newblock Applying fast simulation to find the number of good permutations.
\newblock {\em Cybernetics and Systems Analysis - CYBERN SYST ANAL-ENGL TR},
  43:830--837, 11 2007.

\bibitem{TVZ}
M.~A. Tsfasman, S.~G. Vl\u{a}du\c{t}, and Th. Zink.
\newblock {Modular curves, Shimura curves, and Goppa codes, better than
  Varshamov-Gilbert bound}.
\newblock {\em Mathematische Nachrichten}, 109(1):21--28, 1982.

\bibitem{V91}
Ilan Vardi.
\newblock {\em {Computational recreations in Mathematica}}.
\newblock Addison Wesley, 1991.

\bibitem{XY21}
Chaoping Xing and Chen Yuan.
\newblock Beating the probabilistic lower bound on perfect hashing.
\newblock In D{\'{a}}niel Marx, editor, {\em Proceedings of the 2021 {ACM-SIAM}
  Symposium on Discrete Algorithms, {SODA} 2021, Virtual Conference, January 10
  - 13, 2021}, pages 33--41. {SIAM}, 2021.

\end{thebibliography}

\end{document}